\documentclass[12pt]{article}

\usepackage{amssymb,amsmath,amsfonts,amsthm}
\usepackage{graphicx}
\usepackage{bbm}

\newtheorem{proposition}{Proposition}

\renewcommand{\b}[1]{\mathbf{#1}}
\newcommand{\bi}[1]{\boldsymbol{#1}}
\renewcommand{\c}[1]{\mathcal{#1}}

\renewcommand{\r}[1]{\mathrm{#1}}
\newcommand{\s}[1]{\mathsf{#1}}

\newcommand{\Cx}{\mathbbm{C}}

\newcommand{\idty}{\mathbbm{1}}
\DeclareMathOperator{\id}{id}

\DeclareMathOperator*{\tr}{Tr}
\DeclareMathOperator*{\map}{map}
\newcommand{\<}{\langle}
\renewcommand{\>}{\rangle}
\providecommand{\abs}[1]{|#1|}
\providecommand{\norm}[1]{\Vert #1 \Vert}

\begin{document}

\begin{center}
{\LARGE  Entropic characterization of quantum operations} \\[12pt]
W.~Roga$^1$, M.~Fannes$^2$ and K.~{\.Z}yczkowski$^{1,3}$
\end{center}

\medskip
$^1$ Instytut Fizyki im.~Smoluchowskiego,
Uniwersytet Jagiello{\'n}ski,
PL-30-059 Krak{\'o}w, Poland \\
$^2$ Instituut voor Theoretische Fysica,
Universiteit Leuven, B-3001 Leuven, Belgium \\
$^3$Centrum Fizyki Teoretycznej, Polska Akademia Nauk,
PL-02-668 Warszawa, Poland

\bigskip
\textbf{Abstract:}

We investigate decoherence induced by a quantum channel in terms of minimal output entropy and of map entropy. The latter is the von~Neumann entropy of the Jamio{\l}kowski state of the channel. Both quantities admit $q$-Renyi versions. We prove additivity of the map entropy for all $q$. For the case $q=2$, we show that the depolarizing channel has the smallest map entropy among all channels with a given minimal output Renyi entropy of order two. This allows us to characterize pairs of channels such that the output entropy of their tensor product acting on a maximally entangled input state is larger than the sum of the minimal output entropies of the individual channels. We conjecture that for any channel $\Phi_1$ acting on a finite dimensional system there exists a class of channels $\Phi_2$ sufficiently close to a unitary map such that additivity of minimal output entropy for $\Psi_1 \otimes \Psi_2$ holds.


\section{Introduction}

In quantum information any experimentally realizable set-up that processes states of an $n$-level system is modelled by a quantum operation, also called quantum channel. This is a completely positive affine transformation of the state space. The set of quantum operations has a real dimension $n^2(n^2-1)$ and its structure is far from trivial. Even for the simplest case $n=2$ the structure of the $12$ dimensional convex set of qubit operations is only partially understood~\cite{RSW02}.

The information encoded in a given quantum state is quantified by its von~Neumann entropy or by some similar quantity such as its Renyi entropy of order $q$. The randomizing action of a given quantum channel $\Phi$ can then be characterized by the minimal output entropy $\s S^{\r{min}}_q(\Phi)$: this is the minimal Renyi entropy of order $q$ of an output state of the channel where the minimization is over the entire set of quantum input states. Finding out whether the minimal output entropy is additive with respect to the tensor product of channels was considered to be one of the key questions of quantum information theory. Although Hastings~\cite{hastings} recently showed that in general additivity does not hold, finding explicit counterexamples in low dimensions is still an open problem. An even more relevant question is to specify classes of maps for which additivity holds~\cite{king}.

The decoherence induced in an $n$-level system by a channel may alternatively be characterized by the map entropy $\s S^{\r{map}}(\Phi)$. This quantity, defined as the entropy of the corresponding Jamio{\l}kowski state~\cite{zyczkowski}, varies from zero for a unitary channel to $2\log n$ for the completely depolarizing channel. The entropy of a coarse graining channel with respect to a given basis, $\Phi_{CG}(\rho)={\rm diag}\,(\rho)$, is equal to $\log n$. 
If two quantum maps are close in the sense that the trace distance between the corresponding states is small, then they have similar map entropies \cite{F73}.
The map entropy is easier to determine than the minimal output entropy as there is no minimization to be performed.

The aim of this work is to investigate links between both entropic characterizations of quantum maps. We prove additivity of the map entropy with respect to the tensor product and generalize this result to arbitrary Renyi entropies. To establish relations between the minimal output entropy and the map entropy we investigate the structure of the set of all quantum operations projected onto the plane $\bigr(\s S^{\r{map}}(\Phi), \s S^{\r{min}}(\Phi)\bigl)$. For qubit channels we find the boundaries of this projection and obtain in this way bounds between both quantities.

For the Renyi entropy of order two we show that, for any dimension $n$, the upper boundary of this projection corresponds to the family of depolarizing channels. As for these channels both entropies are explicitly known we obtain inequalities between $\s S^{\r{map}}_2(\Phi)$ and $\s S^{\r{min}}_2(\Phi)$. Applying these results to composite channels and using the additivity of the map entropy we prove a bound for the output entropy of a composite channel minimized over the set of maximally entangled states. This allows us to conjecture that for any two quantum channels of sufficiently different degree of decoherence, e.g.\ $\s S^{\r{map}}(\Phi_1) \gg \s S^{\r{map}}(\Phi_2)$, the minimal output entropy of their product is additive, $\s S^{\r{min}}(\Phi_1 \otimes \Phi_2) = \s S^{\r{min}}(\Phi_1) + \s S^{\r{min}}(\Phi_2)$.

This paper is organized as follows. In Section~2 we introduce some notation and necessary concepts. Some properties of the map entropy, including its additivity with respect to the tensor product, are discussed in Section~3. In Section~4 we derive bounds between the minimal output entropy and the map entropy and we characterize sets of maps for which the additivity of the minimal output entropy can be conjectured. The case of qubit maps is treated in Section~5 where the projection of the entire set of bistochastic quantum operations onto the plane spanned by both entropies is worked out. Some auxiliary material concerning properties of qubit maps is presented in the Appendix.

\section{Quantum channels and their entropies}

A quantum state of an $n$-level system can be identified with a \emph{density matrix} $\rho$ of dimension $n$, i.e.\ a positive definite and normalized matrix:
\begin{equation}
\c D_{n} = \bigl\{ \rho: \Cx^n \to \Cx^n \mid \rho\ge 0,\ \tr\rho = 1 \bigr\}.
\end{equation}

A \emph{quantum operation} or \emph{quantum channel} describes a discrete evolution of the quantum states, it is a linear map $\Phi: \c D_n \to \c D_n$ that is \emph{trace preserving} $(\tr \Phi(\rho) = \tr \rho)$ and \emph{completely positive}. Complete positivity means that the map $\Phi \otimes \id_m$ transforms a positive operator into a positive operator for every dimension $m$ of the extended space. Kraus's theorem~\cite{kraus} says that a map is completely positive if and only if it is of the form $\Phi(\rho) = \sum_{\alpha=1}^r K_\alpha \rho K_\alpha^\dagger$. The trace preserving property is equivalent with $\sum_{\alpha=1}^r K_\alpha^\dagger K_\alpha = \idty$.

The \emph{Jamio{\l}kowski isomorphism}~\cite{jam} represents a quantum map on $\c D_n$ by a state in $\c D_{n^2}$:
\begin{equation}
\sigma^\Phi := \bigl( \Phi\otimes\id \bigr) (|\phi_+\> \<\phi_+|) \enskip\text{with}\enskip |\phi_+\> = \frac{1}{\sqrt n}\, \sum_{i=1}^n |i\> \otimes |i\> = \frac{1}{\sqrt n}\, \sum_{i=1}^n |ii\>.
\label{dyn}
\end{equation}
The matrix $D^\Phi := n\sigma^\Phi$, acting on the doubled space $\c H_A \otimes \c H_B$, is called \emph{dynamical matrix} or \emph{Choi matrix}~\cite{choi}. Positivity of the Choi matrix is equivalent with complete positivity of the corresponding channel $\Phi$, while the partial trace condition $\tr_B D^\Phi = \idty$ is equivalent with $\Phi$ preserving the trace. The rank of the Choi matrix $D^{\Phi}$ is equal to the minimal number of terms needed in a Kraus decomposition. This number is also called the Kraus rank of $\Phi$.

The \emph{Renyi entropy} of order $q$ of a state $\rho$ is defined by
\begin{equation}
\s S_q(\rho) := \frac{1}{1-q}\, \log \tr \rho^q.
\label{ren}
\end{equation}
In the limit $q \to 1$, the Renyi entropy tends to the von~Neumann entropy
\begin{equation}
\s S(\rho) = \lim_{q \to 1} \s S_q(\rho) = -\tr \rho \log \rho.
\end{equation}
For any map $\Phi$ acting on the set $\c D_n$ of quantum states one introduces the \emph{minimum output entropy},
\begin{equation}
\s S^{\r{min}}_q(\Phi) := \min_\rho\ \s S_q\bigl(\Phi(\rho)\bigr),
\end{equation}
where the minimum is taken over all states in $\c D_n$. The interesting question then arises whether the minimal output entropy of the tensor product of two quantum operations is equal to the sum of minimal output entropies of these operations~\cite{king}. The additivity of minimal output entropy is equivalent to the additivity of channel capacity~\cite{shor}. For some special classes of quantum operations additivity of minimal output entropy holds but it is known that it fails in general. The first proof by Hastings~\cite{hastings} was based on random operations acting on high dimensional state spaces and was not constructive. Later some concrete counterexamples to additivity were presented in~\cite{horodecki}.

Another characteristic of the decoherent behaviour of a quantum channel is the \emph{map entropy} of the channel which is the entropy of the rescaled dynamical matrix $\sigma^\Phi$~\cite{zyczkowski}:
\begin{equation}
\s S^{\r{map}}_q(\Phi) := \frac{1}{1-q}\, \log \tr(\sigma^\Phi)^q.
\end{equation}
The map entropy is equal to $0$ if and only if $\Phi$ is a unitary operation. It reaches its maximum, $2 \log n$, at the maximally depolarizing channel $\Phi_*$ which transforms any initial state into the maximally mixed state $\rho_* := \frac{1}{n}\, \idty$. The map entropy was considered earlier in the context of quantum capacity: the quantum capacity of a bistochastic qubit channel of Kraus rank two was shown to be equal to its map entropy~\cite{VV03}. Several properties of this entropy were recently discussed in~\cite{WR, MZ08}. This quantity can be used to bound the Holevo information of output states of a measurement apparatus~\cite{WRPRL} defined by the Kraus operators of a quantum channel. The map entropy is as a special instance of the exchange entropy: it is the entropy of the environment, initially in a pure state, after an action of the quantum operation on the maximally mixed state. For bistochastic channels, i.e.\ channels preserving the maximally mixed state, the map entropy is subadditive with respect to concatenation~\cite{WR}:
\begin{equation}
\s S^{\r{map}}(\Phi_2 \circ \Phi_1)  \le \s S^{\r{map}}(\Phi_1) + \s S^{\r{map}}(\Phi_2).
\end{equation}
A generalization of this relation to general quantum maps was also found. Further properties of the map entropy and its relation to the minimal output entropy are discussed in the subsequent sections.

\section{Properties of the map entropy}

\subsection{Extremal values of entropic characteristics}
\label{sec2}

Further on depolarizing channels play a distinguished role, they form a one parameter family of quantum operations $\Lambda_n$ on $\c D_n$~\cite{kingdepol}:
\begin{equation}
\Lambda_n(\rho) := \lambda \rho + (1-\lambda)\, \tfrac{1}{n}\, \idty \enskip\text{where}\enskip \lambda \in [-\frac{1}{n^2-1}, 1].
\label{depol}
\end{equation}
The constraint on $\lambda$ ensures the complete positivity of $\Lambda_n$. The minimal Renyi output entropy of such a channel can be computed explicitly~\cite{kingdepol} by considering the image of an arbitrary pure state:
\begin{equation}
\s S^{\r{min}}_2(\Lambda_n) = -\log \Bigl( \frac{1+(n-1)\lambda^2}{n}\Bigr).
\label{smin}
\end{equation}
According to (\ref{dyn}) the normalized dynamical matrix of a depolarizing channel reads
\begin{equation}
\sigma^{\Lambda_n} = \frac{1}{n}\Bigl(\, \sum_{ij} \lambda|i\> \<j| \otimes |i\> \<j| + \frac{1-\lambda}{n}\, \delta_{ij}\, \idty \otimes |i\> \<j| \Bigr),
\end{equation}
where $\delta_{ij}$ denotes the Kronecker delta. Therefore the map Renyi entropy of order two is given by
\begin{equation}
\s S^{\r{map}}_2(\Lambda_n) = -\log \Bigl( \frac{1+(n^2-1)\lambda^2}{n^2}\Bigr).
\label{smap}
\end{equation}
Note that both the ranges of values of the minimal Renyi output entropy and of the map Renyi output entropy coincide with the full ranges that such entropies can attain.

\begin{proposition}
\label{prop:depol}
Among all channels with a given minimal Renyi output entropy of order two the depolarizing channel has the smallest map Renyi entropy.
\end{proposition}

\begin{proof}
Putting $\s S^{\r{min}}_2(\Lambda_n) = - \log(1 - \epsilon)$
\begin{equation}
\s S^{\r{map}}_2(\Lambda_n) = -\log \Bigl( 1 - \frac{\epsilon(n+1)}{n} \Bigr).
\end{equation}
The aim is to prove that the map entropy of a quantum operation $\Phi$ on $\c D_n$  is not less than the map entropy of a depolarizing channel with the same minimal Renyi output entropy. Equivalently we want to show that
\begin{equation}
\tr \Bigl( \Phi(|\varphi\> \<\varphi|) \Bigr)^2 \le 1-\epsilon \ \Longrightarrow\  \tr \bigl(\sigma^\Phi \bigr)^2 \le 1 - \frac{\epsilon(n+1)}{n},
\label{imp}
\end{equation}
where $D^\Phi = n \sigma^\Phi$ is the Choi matrix of $\Phi$.

Using a Kraus decomposition
\begin{equation}
\Phi(\rho) = \sum_\alpha K_\alpha \rho K_\alpha^\dagger,\ \sum_\alpha K_\alpha^\dagger K_\alpha = \idty
\end{equation}
we find
\begin{equation}
\tr \Phi(|\varphi\> \<\varphi|)^2 = \sum_{\alpha,\beta} \< \varphi \otimes \varphi \,,\,  K_\alpha^\dagger K_\beta \otimes K_\beta^\dagger K_\alpha\, \varphi \otimes \varphi \>
\end{equation}
and
\begin{equation}
\tr \bigl( \sigma^\Phi \bigr)^2 = \frac{1}{n^2}\, \sum_{\alpha,\beta} \bigl| \tr K_\alpha K_\beta^\dagger \bigr|^2.
\end{equation}

Now we use the following result: let $\mu$ be the Haar measure on the unitary matrices $\c U_n$ of dimension $n$ and let $A$ be a matrix of dimension $n^2$, then~\cite{thir}
\begin{equation}
\int_{\c U_n} \!\mu(dU)\, U \otimes U\, A\, U^\dagger \otimes U^\dagger = \Bigl( \frac{\tr A}{n^2-1} - \frac{\tr AF}{n(n^2-1)} \Bigr)\, \idty - \Bigl( \frac{\tr A}{n(n^2-1)} - \frac{\tr AF}{n^2-1} \Bigr)\, F.
\end{equation}
Here $F$ denotes the swap operation: $F(\varphi \otimes \psi) = \psi \otimes \varphi$. We apply this result to find
\begin{equation}
\int_{\c U_n} \!\mu(dU)\, \< U\varphi \otimes U\varphi \,,\, A\, (U\varphi \otimes U\varphi) \> = \frac{1}{n(n+1)}\, ( \tr A + \tr AF).
\end{equation}
This allows us to write the inequality
\begin{equation}
\frac{1}{n(n+1)}\, \sum_{\alpha,\beta} \bigl( \bigl| \tr K_\alpha^\dagger K_\beta \bigr|^2 + \tr K_\alpha K_\alpha^\dagger K_\beta K_\beta^\dagger \bigr) \le 1 - \epsilon.
\label{aver}
\end{equation}
Now, by Schwarz's inequality for the Hilbert-Schmidt inner product
\begin{equation}
n^2 = \Bigl( \tr \sum_\alpha K_\alpha K_\alpha^\dagger \Bigr)^2 \le n\, \sum_{\alpha,\beta} \tr K_\alpha K_\alpha^\dagger K_\beta K_\beta^\dagger
\end{equation}
and~(\ref{aver}) implies
\begin{equation}
\frac{1}{n^2}\, \sum_{\alpha,\beta} \bigl|\tr K_\alpha^\dagger K_\beta \bigr|^2  \le 1 - \frac{\epsilon(n+1)}{n},
\end{equation}
which proves~(\ref{imp}).
\end{proof}

For any dimension $n \ge 2$ the minimal second Renyi output entropy of a depolarizing channel is a continuous, monotonically increasing, and concave function of its map entropy on the entire domain of the map entropy:
\begin{equation}
\s S_2^{\r{min}} \Big( \s S_2^{\r{map}}(\Lambda_n) \Big) = -\log\Big( \frac{1+n\r e^{-\s S_2^{\r{map}}(\Lambda_n)}}{n+1} \Big).
\end{equation}
This implies that the following statement is also true: among all maps of a same map entropy of order two the depolarizing channel has the largest minimal output entropy. In other words, representing in the $\big( \s S_2^{\r{map}}(\Phi),\s S_2^{\r{min}}(\Phi) \big)$-plane the set of all quantum operations, there are no points above the line corresponding to the depolarizing channels. This result holds in any dimension.

\subsection{Additivity of the map entropy}
\label{sec1}

\begin{proposition}
\label{addit}
Let $\Phi_1$ and $\Phi_2$ be trace preserving, completely positive maps. For any $q\ge 0$ the Renyi map entropy satisfies the additivity relation:
\begin{equation}
\s S^{\r{map}}_q(\Phi_1 \otimes \Phi_2) = \s S^{\r{map}}_q(\Phi_1) + \s S^{\r{map}}_q(\Phi_2).
\end{equation}
\end{proposition}

\begin{proof}
We show that $D^{\Phi_1 \otimes \Phi_2}$ is unitarily equivalent with $D^{\Phi_1} \otimes D^{\Phi_2}$ from which additivity of the map entropies follows. To do so, it is convenient to equip the $n$-dimensional matrices with the Hilbert-Schmidt inner product
\begin{equation}
\< A \,,\, B \>_{\b h} := \tr A^\dagger B.
\end{equation}
In this space the matrix units $\bigl\{ |i\> \<j| \,\bigm|\, i,j = 1,2,\ldots,n\bigr\}$ form an orthonormal basis. We use the notation $|i\> \<j| := |ij\>_{\b h}$. A channel $\Phi$ is now represented by a matrix $\hat\Phi$:
\begin{equation}
\<ij \,,\, \hat\Phi\, k\ell\>_{\b h} = \tr  \Bigl( |j\> \<i|\, \Phi(|k\> \<\ell|) \Bigr),
\end{equation}
hence
\begin{equation}
\Phi(|k\> \<\ell|) = \sum_{i,j} \<ij \,,\, \hat\Phi\, k\ell\>_{\b h}\, |i\> \<j|.
\label{rozklad}
\end{equation}
Therefore, the entries of the dynamical matrix~(\ref{dyn}) can be obtained by permuting the entries of the matrix $\hat\Phi$:
\begin{equation}
\<ab \,,\, D^\Phi\, cd\>_{\b h} = \<ac \,,\, \hat\Phi\, bd\>_{\b h}.
\label{tas}
\end{equation}
We define an unnormalized maximally entangled state $|\Psi_+\> := \sum_{i,\ell} |i\ell\> \otimes |i\ell\>$ and compute directly the entries of $D^{\Phi_1 \otimes \Phi_2}$:
\begin{align}
\, \<abcd \,,\, D^{\Phi_1 \otimes \Phi_2}\, efgh \>
&= \<abcd \,,\, \bigl[ (\Phi_1 \otimes \Phi_2) \otimes \id \bigr] \bigl( |\Psi_+\> \<\Psi_+| \bigr)\, efgh\>
\nonumber \\
&= \sum_{i,\ell,j,m} \<abcd \,,\, \bigl[ (\Phi_1 \otimes \Phi_2)(|i\ell\> \<jm|) \otimes |i\ell\> \<jm| \bigr]\, efgh\>.
\label{rez}
\end{align}
Now we use~(\ref{rozklad}) and obtain:
\begin{equation}
\<abcd \,,\, D^{\Phi_1 \otimes \Phi_2}\, efgh\> = \sum_{\alpha,\beta,\gamma,\delta} \<\alpha\beta \,,\, \widehat{\Phi_1}\, ij\>_{\b h}\, \<\gamma\delta \,,\, \widehat{\Phi_2}\, ij\>_{\b h}\, \<abcd \,,\, \alpha\gamma i\ell\>\, \<\beta\delta jm \,,\, efgh\>.
\end{equation}
After summation over the Greek indices we get:
\begin{align}
\<abcd \,,\, D^{\Phi_1 \otimes \Phi_2}\, efgh\>
&= \<ac \,,\, D^{\Phi_1}\, eg\>\, \<bd \,,\, D^{\Phi_2}\, fh\>
\nonumber \\
&= \<acbd \,,\, D^{\Phi_1} \otimes D^{\Phi_2}\, egfh\>.
\end{align}
The matrix $D^{\Phi_1 \otimes\Phi_2}$ is not equal to $D^{\Phi_1} \otimes D^{\Phi_2}$. However, both are related by a unitary permutation matrix which exchanges the second and the third indices: $U = \sum_{a,b,c,d} |abcd\> \<acbd|$. Therefore both matrices have the same spectra and hence the same entropies.
\end{proof}

We present some applications of Propositions~\ref{prop:depol} and~\ref{addit} in the next section.

\section{Implications on minimal output entropy}
\label{sec4}

The additivity conjecture states that sending an entangled state through a product channel $\Phi_1 \otimes \Phi_2$ yields an output state with entropy not less than the smallest output entropy of input states with a  product structure. A counterexample to this conjecture was given e.g.\ in~\cite{hastings}, where estimating the entropy of an output state arising from a maximally entangled input state plays an important role. This convinces us that it is useful to find methods for estimating the output entropy of maximally entangled input states. We use entropic characteristics to provide a typical estimation. We also use the propositions of the previous section to characterize a class of channels for which we conjecture additivity of minimal output entropy.

\begin{proposition}
\label{lind}
For any maximally entangled state $|\psi_+\>$ the following inequality for von~Neumann entropy holds:
\begin{equation}
\bigl| \s S^{\r{map}}(\Phi_1) - \s S^{\r{map}}(\Phi_2) \bigr| \le  S\bigl( (\Phi_1 \otimes \Phi_2) (|\psi_+\> \<\psi_+|) \bigr) \le \s S^{\r{map}}(\Phi_1) + \s S^{\r{map}}(\Phi_2).
\label{lindblad2}
\end{equation}
\end{proposition}

\begin{proof}
Lindblad's inequality~\cite{lindblad} states that
\begin{equation}
\bigl| \s S(\rho) - \s S(\varsigma(\Phi,\rho)) \bigr| \le \s S(\Phi(\rho)) \le \s S(\rho) + \s S(\varsigma(\Phi,\rho)),
\label{lili}
\end{equation}
where the state $\varsigma(\Phi,\rho)$ is the output state of the channel $\Phi\otimes\id$ acting on a purification of $\rho$. The quantity $\s S(\varsigma(\Phi,\rho))$ is called the exchange entropy and does not depend on the chosen purification. We apply Lindblad's inequality to
\begin{equation}
\s S\bigl( (\Phi_1 \otimes \Phi_2) (|\psi_+\> \<\psi_+|) \bigr) = \s S\Bigl( (\Phi_1 \otimes \id) \bigl( (\id \otimes \Phi_2) (|\psi_+\> \<\psi_+|) \bigr) \Bigr).
\end{equation}
Note that by the definition of the dynamical matrix one has 
\begin{equation*}
(\id \otimes \Phi_2) (|\psi_+\> \<\psi_+|) = \sigma^{\Phi_2}=D^{\Phi_2}/n.
\end{equation*}
We get
\begin{equation}
\bigl| \s S^{\r{map}}(\Phi_2) - \s S\bigl( \varsigma(\Phi_1\otimes\id,\, \sigma^{\Phi_2}) \bigr) \bigr| \le
\s S\bigl( (\Phi_1 \otimes \Phi_2) (|\psi_+\> \<\psi_+|) \bigr) \le \s S^{\r{map}}(\Phi_2) + \s S\bigl( \varsigma(\Phi_1\otimes\id,\, \sigma^{\Phi_2}) \bigr).
\end{equation}
The exchange entropy $\s S\bigl( \varsigma(\Phi_1\otimes\id,\, \sigma^{\Phi_2}) \bigr)$ is equal to $\s S\bigl( \varsigma(\Phi_1,\, \tr_2 \sigma^{\Phi_2}) \bigr)$ because a purification of $\sigma^{\Phi_2}$ is a special case of a purification of $\tr_2\sigma^{\Phi_2}$. Because $\Phi_2$ is a trace preserving map $\tr_2\sigma^{\Phi_2} = \rho_*$~\cite{KZ}. Moreover, $\varsigma(\Phi_1,\, \rho_*) = \sigma^{\Phi_1}$. This completes the proof.
\end{proof}

Since Lindblad's inequality~(\ref{lili}) is based on subadditivity of entropy, Proposition~\ref{lind} can be generalized to other entropies which satisfy this property. Renyi entropy of order $2$ is not subadditive in contrast to Tsallis $q$-entropy:
\begin{equation}
\s T_q(\rho) := \frac{1}{(1-q)}\, \bigl( \tr \rho^q - 1 \bigr),
\end{equation}
which is sub-additive for $q>1$~\cite{audenaert}. As the Tsallis 2-entropy is a function of the Renyi 2-entropy, the Tsallis 2-version of the lower bound in~(\ref{lindblad2}) also yields a lower bound for the Renyi 2-entropy $\s S_2$ of a product channel acting on a maximally mixed initial state:
\begin{equation}
-\log \Bigl( 1 - \bigl| \r e^{-\s S_2^{\r{map}}(\Phi_1)} - \r e^{-\s S_2^{\r{map}}(\Phi_2)} \bigr| \Bigr) \le \s S_2\bigl( (\Phi_1\otimes\Phi_2) (|\psi_+\> \<\psi_+|) \bigr).
\label{tsal}
\end{equation}

We can now characterize pairs of channels $\Phi_1$ on $\c D_n$ and $\Phi_2$ on $\c D_m$ for which the maximally entangled state is certainly not the minimizer of the output Renyi entropy $\s S_2$. Although this is not a necessary condition for channels for which the additivity holds, it suggests pairs of maps for which additivity may hold. The maximally entangled state is certainly not the minimizer of $\rho \mapsto \s S_2\bigl( (\Phi_1\otimes\Phi_2) (\rho) \bigr)$ if the lower bound in~(\ref{tsal}) is larger than the minimal output entropy  of a depolarizing channel~$\Lambda_{nm}$ which satisfies $\s S_2^{\r{map}}(\Lambda_{nm}) = S_2^{\r{map}}(\Phi_1\otimes\Phi_2)$. A sufficient, but not necessary, condition on pairs of channels $(\Phi_1, \Phi_2)$ for which a maximally entangled state is not the minimizer of the output entropy can be written using~(\ref{smin}) and~(\ref{smap})
\begin{equation}
1 - \frac{nm+1}{nm}\, \bigl| \r e^{-\s S_2^{\r{map}}(\Phi_1)} - \r e^{-\s S_2^{\r{map}}(\Phi_2)} \bigr| \le \r e^{-\s S_2^{\r{map}}(\Phi_1 \otimes \Phi_2)} = \r e^{-\bigl( \s S_2^{\r{map}}(\Phi_1) + \s S_2^{\r{map}}(\Phi_2) \bigr)}.
\label{ineq:cond}
\end{equation}

\begin{figure}[h]
\centering
\scalebox{0.8}{\includegraphics{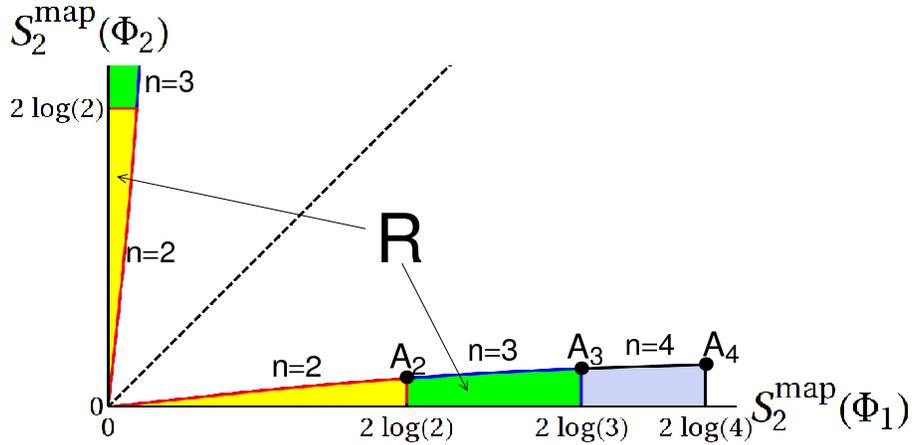}}
\caption{$R$ is the region in the $\bigl( \s S_2^{\r{map}}(\Phi_1), \s S_2^{\r{map}}(\Phi_2) \bigr)$-plane for which additivity of minimal output entropy may hold. Its boundary is determined by inequality~(\ref{ineq:cond}) with $m=n$. It characterizes a class of channels $\Phi_1 \otimes \Phi_2$ for which any maximally entangled input state does not decrease the output entropy $\s S_2$ below the smallest value obtained by states with a tensor product structure. For higher dimensions the allowed region may be enlarged. The dashed line characterizes pairs of complex conjugate channels used in~\cite{hastings} to show violation of additivity of $\s S^{\r{min}}$.}
\label{fig:wykres}
\end{figure}

In Fig.~\ref{fig:wykres} the region $\bigl(\s S_2^{\r{map}}(\Phi_1), \s S_2^{\r{map}}(\Phi_2))$ wherein inequality~(\ref{ineq:cond}) holds is plotted for $m = n = 2,3,4$. For any channel one can choose another one of sufficiently small $\s S_2^{\r{map}}$ to obtain a pair for which no maximally entangled state minimizes the output entropy. For such pairs of channels we may thus conjecture additivity of the minimum output entropy. The map entropy provides only sufficient information to recognize whether two channels belong to this set. The set $R$, for which additivity of $\s S_2^{\r{min}}$ can be conjectured, consists of two regions close the axes and is symmetric with respect to the diagonal. It consists of pairs of maps such that the decoherence induced by one map, as measured by the entropy, is much smaller than the decoherence induced by the other one:
\begin{equation}
\s S_2^{\r{map}}(\Phi_2) \le \alpha_n \s S_2^{\r{map}}(\Phi_1).
\end{equation}
The coefficient
\begin{equation}
\alpha_n := \frac{1}{2\log n}\, \log\Bigl( \frac{n^2(n^2+2)}{n^2(n^2+1)+1} \Bigr)
\end{equation}
is the slope of the line joining the origin with the point $A_n$ from the boundary of $R$ such that $\s S_2^{\r{map}} = 2 \log n$. The counterexamples to additivity used in~\cite{hastings} are conjugated channels, they have therefore a same map entropy and belong to the diagonal $\s S_2^{\r{map}}(\Phi_1) = \s S_2^{\r{map}}(\Phi_2)$ in Fig.~\ref{fig:wykres}. Note that for large $n$ the coefficient $\alpha_n$ tends to zero implying that the maps for which additivity may hold are atypical.

\section{Qubit maps}

Proposition~\ref{prop:depol} determines the upper boundary of the projection of the set of all quantum channels on the $\bigl( \s S_2^{\r{map}}, \s S_2^{\r{min}} \bigr)$-plane. For bistochastic qubit maps the remaining boundaries correspond to quantum maps at the edges of the tetrahedron of bistochastic qubit channels.

\subsection{The asymmetrical tetrahedron of bistochastic qubit maps}

Consider the set of all bistochastic quantum channels on $\c D_2$. Up to two unitary rotations, they are convex combinations of unitary channels determined by Pauli operators and are therefore called \emph{Pauli channels}:
\begin{equation}
\Phi_{\vec{b}}(\rho) = \sum_{i=0}^3 b_i\, \sigma_i\rho\sigma_i,\enskip b_i \ge 0, \enskip\text{and}\enskip \sum_{i=0}^3 b_i = 1.
\label{paulich}
\end{equation}
Here $\{\sigma_i \mid i = 0,1,2,3 \}$ denotes the identity matrix and the three Pauli matrices. This convex set is a tetrahedron and its four vertices correspond to the identity and to three unitary rotations generated by Pauli matrices with respect to three perpendicular axes. The set is shown in Fig.~\ref{fig:czworoscian}~$a)$.

\begin{figure}[h]
\centering
\scalebox{0.6}{\includegraphics{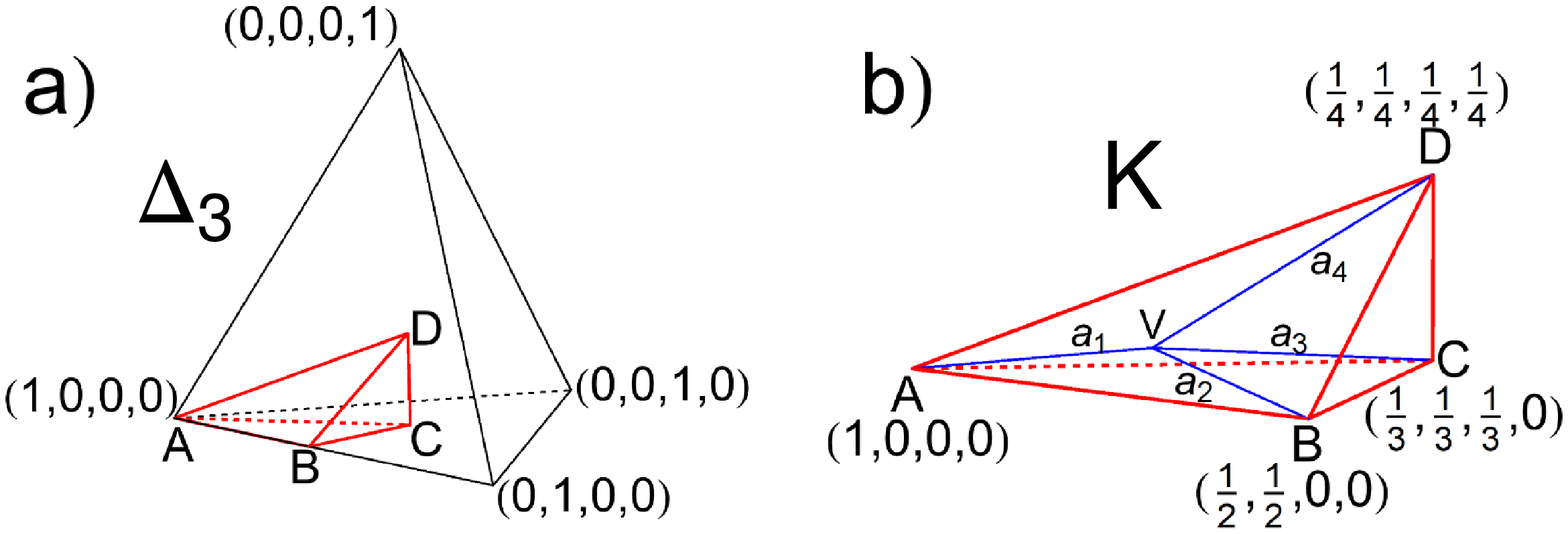}}
\caption{$a)$ The tetrahedron $\Delta_3$ is the set of bistochastic qubit channels. An asymmetric tetrahedron $K$ inside $\Delta_3$ is magnified in panel $b)$. Any point in $K$ is a convex combination~(\ref{eigen}) of the vertices $W_i$ with weights $a_i$.}
\label{fig:czworoscian}
\end{figure}

Using an appropriate permutation of the vertices of the tetrahedron we may restrict our attention to the asymmetric part $K$ of $\Delta_3$, defined as the convex hull of four vectors:
\begin{equation}
\begin{split}
&A = W_1 = (1,0,0,0), \\
&B = W_2 = (\tfrac{1}{2},\tfrac{1}{2},0,0), \\
&C = W_3 = (\tfrac{1}{3},\tfrac{1}{3},\tfrac{1}{3},0),\ \text{and} \\
&D = W_4 = (\tfrac{1}{4},\tfrac{1}{4},\tfrac{1}{4},\tfrac{1}{4}).
\end{split}
\end{equation}
Any point $V\in K$ is a convex combination of the vertices $W_i$
\begin{equation}
V = \sum_i a_i W_i,\enskip a_i \ge 0\ \forall_i, \enskip\text{and}\enskip \sum_i a_i = 1.
\label{eigen}
\end{equation}
The channel corresponding to a point inside $K$ transforms the Bloch ball into an ellipsoid with axes of ordered lengths $\abs{\lambda_1} \le \abs{\lambda_2} \le \abs{\lambda_3} \le 1$, see the Appendix. Comparing~(\ref{eigen}) with~(\ref{uklad}) yields the coefficients $a_i$ in terms of the $\lambda_i$:
\begin{equation}
\begin{split}
&a_1 = \tfrac{1}{2}\, (\lambda_1 + \lambda_2), \\
&a_2 = (\lambda_3 - \lambda_2), \\
&a_3 = \tfrac{3}{2}\, (\lambda_2 - \lambda_1),\ \text{and} \\
&a_4 = (1 + \lambda_1 - \lambda_2 - \lambda_3).
\end{split}
\end{equation}
The extreme points $W_i$ represent quantum maps of different ranks: identity ($a_1=1$, all $\lambda_i=1$), coarse graining ($a_2=1$, $\lambda_3=1$, $\lambda_1=\lambda_2=0$), a depolarizing channel ($a_3=1$, $-\lambda_1=\lambda_2=\lambda_3=\frac{1}{3}$) and the completely depolarizing channel ($a_4=1$, $\lambda_1=\lambda_2=\lambda_3=0$). The convex combinations of these four maps exhaust all possible shapes of ellipsoids which are images of the Bloch ball under bistochastic channels, see the Appendix. Due to the specific choice of $K$ the longest axis of the ellipsoid is parallel to the $z$-axis. The position of the ellipsoid with respect to the axes of the Bloch ball has, however, no influence neither on the minimal output entropy nor on the  map entropy.

The minimal output entropy of the quantum map $\Phi$ corresponding to a given point $V$ is
\begin{equation}
\s S_2^{\r{min}}(\lambda_1,\lambda_2,\lambda_3) = -\log \tfrac{1}{2}\, (1 + \lambda_3^2).
\end{equation}
As $V$ is a vector of eigenvalues of the dynamical matrix $D^\Phi$, see the Appendix, the Renyi entropy of $\Phi$ equals
\begin{equation}
\s S_2^{\r{map}}(\Phi) = -\log \sum_{i=0}^3 \abs{b_i} = -\log \norm{V}^{2}.
\end{equation}
Both these entropies depend only on the values $\lambda_i$. The lines representing the edges of the asymmetric tetrahedron in the $\bigl( \s S_2^{\r{map}}, \s S_2^{\r{min}} \bigr)$-plane are shown in Fig.~\ref{fig:reni2} and we show in the next section that they correspond indeed to the boundaries of the allowed region in the entropy plane $\bigl( \s S_2^{\r{map}}, \s S_2^{\r{min}} \bigr)$.

\subsection{The extremal values of minimal output entropy}

\begin{figure}[h]
\centering
\scalebox{0.6}{\includegraphics{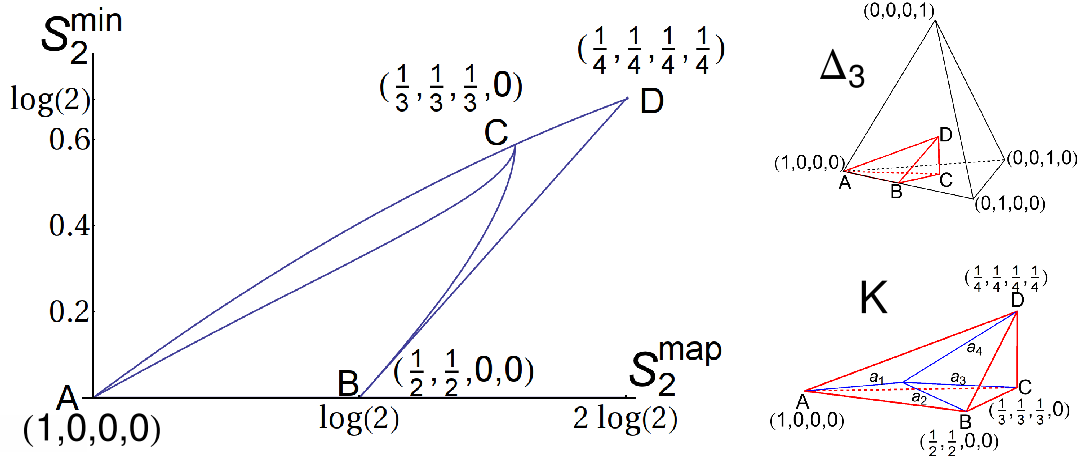}}
\caption{Boundaries of the set $\Delta_3$ of bistochastic qubit channels projected on the $\bigl( \s S_2^{\r{map}}, \s S_2^{\r{min}})$-plane, the tetrahedron $\Delta_3$, and its asymmetric part $K$.}
\label{fig:reni2}
\end{figure}

\begin{proposition}
The boundaries of the set of Pauli channels projected on the $\bigl( \s S_2^{\r{map}}, \s S_2^{\r{min}} \bigr)$-plane correspond to the edges of the asymmetrical tetrahedron $K \subset \Delta_3$.
\end{proposition}

\begin{proof}
Consider figure~\ref{fig:czworoscian}~$b)$. Take $a_3 = 0$ and $a_4 = 0$ so that $a_2 = 1-a_1$ and use~(\ref{minout}) to see that $\s S^{\r{min}}_{2} = 0$. This is the smallest possible value of minimal output entropy. The line $AB$ in Fig.~\ref{fig:reni2}, corresponding to the dephasing channels, describes such maps. The proof that the line $AD$, characterizing depolarizing channels, is a boundary of the set is given in Section~\ref{sec2} in general, not necessarily for qubits or for bistochastic channels. The line $BD$, corresponding to $a_1 = a_3 = 0$ in the tetrahedron, represents classical bistochastic maps, characterized by a diagonal dynamical matrix. All bistochastic qubit channels which have the same minimal output entropy have the same the longest axis $\abs{\lambda_3}$. They are situated on the horizontal line in the $\bigl( \s S_2^{\r{map}}, \s S_2^{\r{min}} \bigr)$ plot. The dynamical matrix of such maps reads:
\begin{equation}
\frac{1}{2}D^\Phi =\sigma^{\Phi}= \frac{1}{4}\, \begin{pmatrix}
&1+\lambda_3  &0 &0 &\lambda_1 + \lambda_2 \\
&0  &1 - \lambda_3 &\lambda_1 - \lambda_2 &0 \\
&0  &\lambda_1 - \lambda_2 &1 - \lambda_3 &0 \\
&\lambda_1 + \lambda_2 &0 &0 &1 + \lambda_3
\end{pmatrix}.
\end{equation}
The dynamical matrix of a classical bistochastic qubit map $\Phi_{\r c}$ of the same minimal output entropy contains only diagonal elements of this matrix $D^{\Phi_{\r c}} = \r{diag}(D^\Phi)$. Due to the majorization theorem the spectrum of a density matrix majorizes its diagonal. Schur concavity of the Renyi entropy $\s S_q$ for $q \ge 1$, see e.g.~\cite{augusiak}, implies that $\s S_q(D^\Phi) \le \s S_q(\r{diag}(D^\Phi))$. Therefore classical bistochastic channels have the greatest map entropy among all bistochastic maps with the same minimal output entropy.  This completes the proof.
\end{proof}

Qubit stochastic channels can occupy also the space on the right from the line $BD$.

\section{Conclusions}

In this work we prove in general additivity of map entropies with respect to the tensor product of channels. We also analyse the relation between two entropic characteristics of a quantum channel: the map entropy $\s S_2^{\r{map}}$ and the minimal output entropy $\s S_2^{\r{min}}$. This approach allows us to distinguish a class of product channels for which additivity of minimal output Renyi entropy of order 2 is conjectured. The relation between the minimal output entropy and the map entropy distinguishes the depolarizing channels as those which form a part of the boundary of the set of all quantum maps projected on the $\bigl( \s S_2^{\r{map}}, \s S_2^{\r{min}})$-plane. The image of the bistochastic qubit channels on this plane was analysed using the asymmetric tetrahedron of Pauli channels.

\begin{figure}[h]
\centering
\scalebox{1}{\includegraphics{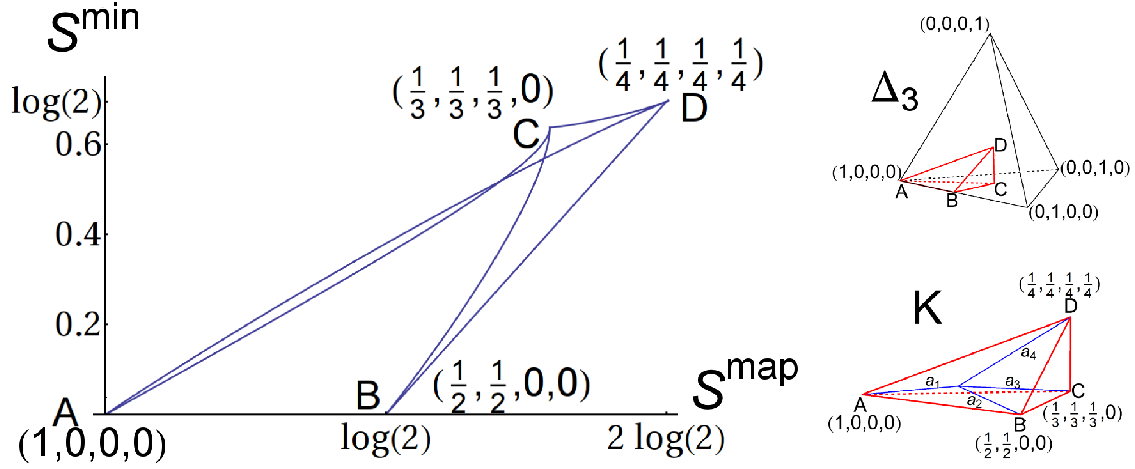}}
\caption{Boundaries of the set of bistochastic qubit channels projected on the plane of von~Neumann entropies $(\s S^{\map},\s S^{\min})$.}
\label{fig:konturzbrylami}
\end{figure}

A similar projection of the set of operations determined by the von~Neumann entropy instead of the second Renyi entropy does not yield a convex set, see Fig.~\ref{fig:konturzbrylami}. Moreover, the family of depolarizing channels does not correspond in this case with the upper boundary of the set. Thus the reasoning in Section~\ref{sec4} about the set of maps for which additivity of minimal output Renyi entropy $\s S_2^{\r{min}}$ is conjectured cannot be directly transferred to the von~Neumann entropy. On the other hand the Renyi entropy is a non-increasing function of the parameter $q$, in particular $\s S(\rho) \ge \s S_2(\rho)$.  Moreover, in  finite dimensions the Renyi entropy depends continuously on its parameter. It is therefore likely that a similar statement holds also for the von Neumann entropy: For two channels $\Phi_1$ and $\Phi_2$, acting on a space of $n$-dimensional density matrices, such that  $\s S^{\r{map}}(\Phi_1) \gg  \s S^{\r{map}}(\Phi_2)$ we  conjecture the additivity of the minimal output entropy, $\s S^{\r{min}}(\Phi_1 \otimes \Phi_2) =  \s S^{\r{min}}(\Phi_1) +  \s S^{\r{min}}(\Phi_1)$.

The above statement also suggests that one should consider in low dimensions two channels with a same map entropy in order to find a counterexample to the additivity conjecture of minimal output entropy. This is precisely the case for Hastings's counterexample~ \cite{hastings} in which a random channel and its conjugate were used.

\section*{ Appendix. Qubit channels}

Any qubit density matrix $\rho$ can be decomposed in the basis of the identity and three Pauli matrices. This decomposition is called the \textit{Bloch representation:}
\begin{equation}
\rho = \tfrac{1}{2}\, (\idty + \vec{\b w} \cdot \vec{\bi\sigma}).
\label{mgest}
\end{equation}
Since a density matrix is Hermitian the Bloch vector $\vec{\b w}$ is real while positivity of $\rho$ is equivalent with $\norm{\vec{\b w}} \le 1$. The set of all such vectors $\vec{\b w}$ is the \textit{Bloch ball}. Therefore one can represent any affine transformation of the qubit states, \textit{a fortiori} a qubit channel, by a $4\times4$ matrix $\Phi$ acting on the extended Bloch vector $(1, \vec{\b w})^{\s T}$. One can choose a basis such that
\begin{equation}
\Phi = \begin{pmatrix}
&1  &0 &0 &0 \\
&t_1 &\lambda_1 &0 &0 \\
&t_2 &0 &\lambda_2 &0 \\
&t_3 &0 &0 &\lambda_3
\end{pmatrix}.
\label{bloch}
\end{equation}
The channel transforms the Bloch ball into the ellipsoid
\begin{equation}
\Bigl( \frac{x-t_1}{\lambda_1} \Bigr)^2 + \Bigl( \frac{y-t_2}{\lambda_2} \Bigr)^2 + \Bigl( \frac{z-t_3}{\lambda_3} \Bigr)^2 \le 1.
\end{equation}
The ellipsoid has three main axes of half lengths $\{ \abs{\lambda_i} \mid i \}$ and its centre is translated with respect to the centre of Bloch ball by the vector $\vec{\b t} = (t_1,t_2,t_3)$. The positivity of the map guarantees that the ellipsoid lies inside the ball.

The corresponding normalized dynamical matrix in Bloch parametrization~(\ref{bloch}) is
\begin{equation}
\frac{1}{2}D^\Phi= \sigma^\Phi = \frac{1}{4}\, \begin{pmatrix}
&1 + \lambda_3 + t_3  &0 &t_1 + \r it_2 &\lambda_1 + \lambda_2 \\
&0  &1 - \lambda_3 + t_3 &\lambda_1 - \lambda_2 &t_1 + \r it_2 \\
&t_1 - \r it_2  &\lambda_1 - \lambda_2 &1 - \lambda_3 - t_3 &0 \\
&\lambda_1 + \lambda_2 &t_1 - \r it_2 &0 &1 + \lambda_3 - t_3
\end{pmatrix}.
\label{dyn2}
\end{equation}
The eigenvalues $v_i$ of the dynamical matrix are connected with the parameters $\lambda_i$ in the Bloch representation of a channel~(\ref{bloch}) by:
\begin{equation}
\begin{split}
&v_1 = \tfrac{1}{4}\, (1 + \lambda_1 + \lambda_2 + \lambda_3), \\
&v_2 = \tfrac{1}{4}\, (1 - \lambda_1 - \lambda_2 + \lambda_3), \\
&v_3 = \tfrac{1}{4}\, (1 - \lambda_1 + \lambda_2 - \lambda_3),\enskip\text{and} \\
&v_4 = \tfrac{1}{4}\, (1 + \lambda_1 - \lambda_2 - \lambda_3).
\end{split}
\label{uklad}
\end{equation}
The vector $\vec{\b v} := (v_1,v_2,v_3,v_4)$ corresponds to the vector $\vec{b}$~from~(\ref{paulich}), see~\cite{KZ}. The minimal output entropy of a bistochastic channel is the minimal entropy of an output state for a pure input state. The output state obtained by acting with the operation~(\ref{bloch}) on a state with Bloch vector $\vec{\b w}$ with $\norm{\vec{\b w }} = 1$ has a Renyi entropy
\begin{equation}
\s S_2^{\r{map}}(\lambda_1, \lambda_2, \lambda_3) := -\log \tfrac{1}{2}\, \bigl( 1 + \lambda_3^2 + w_1^2\, (\lambda_1^2 - \lambda_3^2) + w_2^2\, (\lambda_2^2 - \lambda_3^2) \bigr).
\end{equation}
Because $\abs{\lambda_1} \le \abs{\lambda_2} \le \abs{\lambda_3}$ the coefficients of $w_1$ and $w_2$ are non positive. Hence, $\s S_2^{\r{map}}(\lambda_1, \lambda_2, \lambda_3)$ reaches its minimum when $w_1 = w_2 = 0$ and the minimum Renyi output entropy depends only on the longest axis
\begin{equation}
\s S_2^{\r{min}}(\lambda_1, \lambda_2, \lambda_3) = -\log \tfrac{1}{2}\, \bigl( 1 + \lambda_3^2 \bigr).
\end{equation}
By solving the equations~(\ref{uklad}, \ref{eigen}) we obtain minimal output entropy of $\Phi$ as a function of the weights
\begin{equation}
\s S_2^{\r{min}}(a_1, a_2, a_4) = -\log \tfrac{1}{18}\, \bigl( 9 + (1 + 2a_1 + 2a_2 - a_4)^2 \bigr).
\label{minout}
\end{equation}

\medskip
\textbf{Acknowledgements:}
We acknowledge financial support by the grant number N202 090239 of the Polish Ministry of Science, by the Belgian Interuniversity Attraction Poles Programme P6/02, and by the FWO Vlaanderen project G040710N.

\end{document}